\documentclass[10pt]{amsart}
\UseRawInputEncoding
\usepackage[utf8]{inputenc}
\usepackage{amsfonts}
\usepackage{amssymb}
\usepackage{amsmath,amsthm}
\usepackage{amscd}
\usepackage{graphics}
\usepackage{cancel}
\usepackage{pdfsync}
\usepackage[british]{babel} 
\usepackage{mathrsfs}
\usepackage{enumitem}
\usepackage{stmaryrd}
\usepackage{mathrsfs}
\usepackage{wasysym}
\usepackage{latexsym}
\usepackage[colorlinks,linkcolor=blue,citecolor=blue,urlcolor=blue]{hyperref}

\numberwithin{equation}{section}

\newcommand{\beq}{\begin{equation}}
\newcommand{\eeq}{\end{equation}}
\newcommand{\bea}{\begin{eqnarray}}
\newcommand{\eea}{\end{eqnarray}}
\newcommand{\nn}{\nonumber}
\newcommand\noi{\noindent}
\newcommand{\tbf}{\textbf}

\newcommand{\rd}{\mathrm{d}}

\newcommand{\bk}{\begin{cases}}
\newcommand{\ek}{\end{cases}}
\newcommand{\bs}{\boldsymbol}

\newcommand{\ob}{\boldsymbol{0}}

\newcommand{\cD}{\mathcal{D}}

\newtheorem{definition}{Definition}
\newtheorem{proposition}{Proposition}
\newtheorem{theorem}{Theorem}

\newtheorem{lemma}{Lemma}
\theoremstyle{definition}

\newtheorem{remark}{\textbf{Remark}}

\usepackage{enumerate}
\usepackage{float}
\usepackage{cancel}
\usepackage{hyperref}
\usepackage{mathtools}
\usepackage{afterpage}

\allowdisplaybreaks

\newcommand*{\templ}{\multicolumn{1}{c|}{0}}
\newcommand*{\tempr}{\multicolumn{1}{|c}{0}}

\begin{document}

\title[Generalized Nijenhuis Operators and block-diagonalization]{Generalized Nijenhuis Torsions and block-diagonalization of operator fields}

\author{Daniel Reyes Nozaleda}
\address{Departamento de F\'{\i}sica Te\'{o}rica, Facultad de Ciencias F\'{\i}sicas, Universidad
Complutense de Madrid, 28040 -- Madrid, Spain \\ and Instituto de Ciencias Matem\'aticas, C/ Nicol\'as Cabrera, No 13--15, 28049 Madrid, Spain}
\email{danreyes@ucm.es, daniel.reyes@icmat.es}
\author{Piergiulio Tempesta}
\address{Departamento de F\'{\i}sica Te\'{o}rica, Facultad de Ciencias F\'{\i}sicas, Universidad
Complutense de Madrid, 28040 -- Madrid, Spain \\  and Instituto de Ciencias Matem\'aticas, C/ Nicol\'as Cabrera, No 13--15, 28049 Madrid, Spain}
\email{piergiulio.tempesta@icmat.es, ptempest@ucm.es}
\author{Giorgio Tondo}
\address{Dipartimento di Matematica e Geoscienze, Universit\`a  degli Studi di Trieste,
piaz.le Europa 1, I--34127 Trieste, Italy.}
\email{tondo@units.it}


\date{May 15, 2022}



\begin{abstract}
The theory of generalized Nijenhuis torsions, which extends the classical notions due to Nijenhuis and Haantjes, offers new tools for the study of normal forms of operator fields. We propose a general result ensuring that, given a family of commuting operator fields whose generalized Nijenhuis torsion of level $l$ vanishes, there exists a local chart where all operators can be simultaneously block-diagonalized. We also introduce the notion of generalized Haantjes algebra, consisting of operators with a vanishing higher-level torsion, as a new algebraic structure naturally generalizing standard Haantjes algebras. 
\end{abstract}

\maketitle

\tableofcontents

\section{Introduction}

The purpose of this article is to establish a new geometric setting for the block-diagonalization of operator fields on a differentiable manifold.

The problem of finding normal forms of (1,1)-tensor fields on a differentiable manifold $M$ is relevant in many theoretical and applicative contexts. Indeed, it is useful in the determination of a reduced form for a system of partial differential equations \cite{BogCMP,BogJMP,BogCMP2007} and in the theory of hydrodynamic-type systems \cite{FeMa,FeKhu}. Also, it finds applications in classical mechanics, since it may allow to construct separation variables for completely integrable Hamiltonian systems \cite{CR2019}. 

In this article, we address this problem with a tensorial approach, based on the notion of Nijenhuis and Haantjes tensors and their generalizations, recently introduced in \cite{TT2022CMP}. 

The interest in the geometry of Nijenhuis and Haantjes tensors recently  has considerably increased. Indeed, new applications have been found in many different contexts: for instance, the characterization
of integrable chains of partial differential equations of hydrodynamic type (see e.g. \cite{FeMa}, \cite{FeKhu}) and their integrable reductions \cite{FP2021}, as well as the study of topological field theories \cite{MGall17}. Also, Nijenhuis geometry, which plays a crucial role in the theory of almost complex structures \cite{NNAM1957}, has been extended to contexts as the study of analytic matrix functions, linearisation theory, operator algebras \cite{BKM2022AIM}, 
etc.

In \cite{TT2021JGP}, two of the present authors proposed the notion of Haantjes algebra. It has been shown that for an Abelian algebra of semisimple operator fields with vanishing Haantjes torsion there exist local coordinate charts where the simultaneous diagonalization of all the operators of the algebra is allowed.  
In  \cite{TT2022AMPA}, \cite{RTT2022CNS}, the notion of $\omega \mathscr{H}$  manifolds (namely symplectic manifolds endowed with an algebra of Haantjes operators),  has been proposed as a natural setting for the formulation of the theory of finite-dimensional integrable Hamiltonian  systems. The class of $P\mathscr{H}$ manifolds has been introduced in \cite{T2018TMP}.

A new, infinite family of higher-order torsions of level $m$, generalizing both the Nijenhuis and the Haantjes torsions, have been defined and discussed in \cite{TT2022CMP}. These torsions can also be derived from a family of higher-order Haantjes brackets, which generalize the Fr\"olicher-Nijenhuis (FN) one. As the FN bracket is directly related with the Nijenhuis torsion, so the higher Haantjes brackets are related with our higher-order torsions. Operator fields with a vanishing generalized torsion will be said to be generalized Nijenhuis operators. 

We remind that a different formulation of the theory of generalized Nijenhuis torsions has been proposed in  \cite{KS2019}.

Besides, a generalization of the classical Haantjes theorem was proposed. Precisely, it has been proved that, given an operator field $\boldsymbol{A}:\mathfrak{X}(M)\rightarrow \mathfrak{X}(M)$, the vanishing of any  generalized Nijenhuis torsion is sufficient to guarantee the \textit{mutual integrability} of the eigen-distributions of $\bs{A}$, namely the integrability of each of them and of all of their direct sums. 

From a theoretical point of view, our main result is a theorem stating that for a set of commuting operators on a differentiable manifold, if at least one of them is a generalized Nijenhuis operator, then  there exists a local coordinate chart where they take simultaneously a block-diagonal form. 

This theorem is further refined by requiring that all of the commuting operators are generalized Nijenhuis operators of a (common) level $m$. In this case, on a local coordinate chart, all of them take a block-diagonal form with finer blocks. 

These results naturally suggest the notion of \textit{generalized Haantjes algebra}. It consists of a set of generalized Nijenhuis operators forming a $C^{\infty}$ module; besides, this set is closed under the product of operators. 

 

We will discuss in detail the properties of the class of cyclic generalized Haantjes algebras. They are generated by the independent powers of a given generalized Nijenhuis tensor. We shall also prove that any polynomial $P(\bs{A})$ in an operator $\bs{A}$ such that $\mathcal{\tau}^{(n)}_{\bs{A}}(X,Y)=0$, whose coefficients are functions in $C^{\infty}(M)$, satisfies  $\mathcal{\tau}^{(n)}_{P(\bs{A})}(X,Y)=0$. This property allows us to construct generalized Haantjes algebras in a very natural and direct way. We shall present some nontrivial examples of cyclic generalized Haantjes algebras of level three and four. In particular, we will show how to construct local coordinate charts allowing the simultaneous block-diagonalization of all the operators forming these algebras. 

The approach proposed in this work offers a new perspective concerning the classical problem of the normal form of operator fields. We remind that a relevant contribution to this problem, in the spirit of tensor analysis, was given by Bogoyavlenskij in \cite{BogJMP,BogCMP2007}. In particular, necessary and sufficient conditions were proposed for the existence of local charts were a given operator acquires a block-diagonal form. The main difference with respect to those results is that we solve the problem for a \textit{full family} of commuting operators. At the same time, no  knowledge \textit{a priori} of the eigen-distributions of the given operators is required. This aspect is crucial, since the spectral analysis of operator fields becomes rapidly intractable by increasing the dimension of the underlying differentiable manifold. Thus, we offer sufficient tensor conditions for the simultaneous block-diagonalization of a family of operators, which can be easily checked with the aid of computer algebra, without entering the construction of their eigen-distributions.
Of course, for the explicit construction of the local chart where the simultaneous block-diagonalization of the family of operators considered takes place, it is necessary to enter into their spectral analysis. However, our result can be interpreted as a tensor test, ensuring \textit{a priori} the block-diagonalizability of the full family. Once this property is ascertained for a concrete family, then one can proceed in the construction of the local chart.


A thorough study of the rich algebraic properties of the generalized Haantjes algebras represents an interesting research perspective, that will be developed in a forthcoming work.

The paper is organized as follows.   In Section 2, the theory of Nijenhuis and Haantjes tensors  is briefly reviewed. In Section 3, the concept of generalized Nijenhuis operator is discussed. In Section 4, the new notion of generalized Haantjes algebra is presented. In Section 5, the main theorems of the work, concerning simultaneous block-diagonalization of non-semisimple operator fields, are proved. In Sections 6 and 7, examples of generalized Haantjes algebras are discussed and local charts ensuring simultaneous block-diagonalization are explicitly constructed.

\section{Preliminaries on the Nijenhuis  and Haantjes geometry}
\label{sec:1}
In this section, we shall review some basic notions concerning the geometry of Nijenhuis or Haantjes torsions, following the original papers \cite{Haa1955,Nij1951}  and the related ones
\cite{Nij1955,FN1956}. Here  we shall focus only on the aspects of the theory which are relevant for the subsequent discussion.

Let $M$  be a differentiable manifold of dimension $n$, $\mathfrak{X}(M)$ the Lie algebra of vector fields on $M$ and  $\boldsymbol{A}:\mathfrak{X}(M)\rightarrow \mathfrak{X}(M)$ be a smooth $(1,1)$-tensor field. In the following, all tensors will be considered to be smooth.
\begin{definition} \label{def:N}
The
 \textit{Nijenhuis torsion} of $\boldsymbol{A}$ is the the vector-valued $2$-form defined by
\begin{equation} \label{eq:Ntorsion}
\tau_ {\boldsymbol{A}} (X,Y):=\boldsymbol{A}^2[X,Y] +[\boldsymbol{A}X,\boldsymbol{A}Y]-\boldsymbol{A}\Big([X,\boldsymbol{A}Y]+[\boldsymbol{A}X,Y]\Big),
\end{equation}
where $X,Y \in \mathfrak{X}(M)$ and $[ \ , \ ]$ denotes the commutator of two vector fields.
\end{definition}
\begin{definition} \label{def:H}
 \noi The \textit{Haantjes torsion} associated with $\boldsymbol{A}$ is the vector-valued $2$-form defined by
\begin{equation} \label{eq:Haan}
\mathcal{H}_{\boldsymbol{A}}(X,Y):=\boldsymbol{A}^2 \tau_{\boldsymbol{A}}(X,Y)+\tau_{\boldsymbol{A}}(\boldsymbol{A}X,\boldsymbol{A}Y)-\boldsymbol{A}\Big(\tau_{\boldsymbol{A}}(X,\boldsymbol{A}Y)+\tau_{\boldsymbol{A}}(\boldsymbol{A}X,Y)\Big).
\end{equation}
\end{definition}

The  Haantjes (Nijenhuis) vanishing condition inspires the following definition.
\begin{definition}
A Haantjes (Nijenhuis)   operator is a (1,1)-tensor field  whose  Haantjes (Nijenhuis) torsion identically vanishes.
\end{definition}

A simple, relevant case of Haantjes operator is that of a tensor field   $\boldsymbol{A}$  which takes a diagonal form in a local chart $\boldsymbol{x}=(x^1,\ldots,x^n)$:
\begin{equation}
\boldsymbol{A}(\boldsymbol{x})=\sum _{i=1}^n \lambda_{i }(\boldsymbol{x}) \frac{\partial}{\partial x^i}\otimes \rd x^i \ , \label{eq:Ldiagonal}
 \end{equation}
 where $\lambda_{i }(\boldsymbol{x}):=\lambda^{i}_{i}(\boldsymbol{x})$ are the eigenvalue fields of $\boldsymbol{A}$ and  $\left(\frac{\partial}{\partial x^1},\ldots, \frac{\partial}{\partial x^n}\right) $ are the fields forming the so called \textit{natural frame} associated with the local chart $(x^{1},\ldots,x^{n})$. As is well known, the Haantjes torsion of the diagonal operator \eqref{eq:Ldiagonal} vanishes. 
 
We also recall that two frames  $\{X_1,\ldots,X_n\}$ and $\{Y_1,\ldots,Y_n\}$ are said to be equivalent if $n$ nowhere vanishing smooth
functions $f_i$ exist, such that
\[
 X_i= f_i(\boldsymbol{x}) Y_i \ , \qquad\qquad i=1,\ldots,n \ .
\]
 \begin{definition}\label{def:Iframe}
  An \emph{integrable} frame  is a reference frame equivalent to a natural frame.
\end{definition}

%
%

It is interesting to observe that the algebraic properties of Haantjes operators are richer that those of Nijenhuis operators.  Hereafter,  $\boldsymbol{I}: \mathfrak{X}(M)\rightarrow \mathfrak{X}(M)$ will denote the identity operator.
A  useful result is the following.
\begin{proposition}  \label{pr:fL} \cite{BogCMP}.
Let  $\boldsymbol{A}$ be a (1,1)-tensor field. The following identity holds:
\begin{equation} \label{eq:LtorsionLocal}
\mathcal{H}_{f \boldsymbol{I}+g \boldsymbol{A}}(X,Y)=g^4\, \mathcal{H}_{ \boldsymbol{A}}(X,Y),
\end{equation}
where $f,g:M \rightarrow \mathbb{R}$ are $C^\infty(M)$ functions.
\end{proposition}
\begin{proof}
See Proposition 1, p. 255 of \cite{BogCMP}.
\end{proof}
\noi Interestingly enough, such a property does not hold in the case of a Nijenhuis operator.

\par
\noi Many more examples of Haantjes operators, relevant in classical mechanics and in Riemannian geometry, can be found for instance in the works \cite{RTT2022CNS}--\cite{TT2016SIGMA}.

\section{The generalized Nijenhuis tensors and block-diagonalization}

In this section, for the sake of clarity, we shall briefly review some of the main algebro-geometric properties of the new class of generalized Nijenhuis tensors introduced in  \cite{TT2021JGP} and further studied in \cite{TT2022CMP}.

\begin{definition} \label{df:mtorsion}
Let $\boldsymbol{A}:\mathfrak{X}(M)\rightarrow \mathfrak{X}(M)$ be a (1,1)-tensor field. The generalized Nijenhuis torsion of $\boldsymbol{A}$ of level $m$, for each integer $m\geq 1$, is the  skew-symmetric (1,2)-tensor field defined by
\bea \label{GNTn}
\nn \mathcal{\tau}^{(m)}_{\boldsymbol{A}}(X,Y)=\boldsymbol{A}^2\mathcal{\tau}^{(m-1)}_{\boldsymbol{A}}(X,Y)+
\mathcal{\tau}^{(m-1)}_{\boldsymbol{A}}(\boldsymbol{A}X,\boldsymbol{A}Y) - 
\boldsymbol{A}\Big(\mathcal{\tau}^{(m-1)}_{\boldsymbol{A}}(X,\boldsymbol{A}Y)+\mathcal{\tau}^{(m-1)}_{\boldsymbol{A}}(\boldsymbol{A}X,Y)\Big), \\ \nn \quad X,Y \in \mathfrak{X}(M) \ . \eea
\beq
\eeq
Here the notation $\tau_{\boldsymbol{A}}^{(0)}(X,Y):= [X,Y]$,  $\tau_{\boldsymbol{A}}^{(1)}(X,Y):=\tau_{\boldsymbol{A}}(X,Y)$
and $\tau_{\boldsymbol{A}}^{(2)}(X,Y):=\mathcal{H}_{\boldsymbol{A}}(X,Y)$ is adopted.
\end{definition}

\begin{definition}
Let $\boldsymbol{A}:\mathfrak{X}(M)\rightarrow \mathfrak{X}(M)$ be a (1,1)-tensor field. If $\mathcal{\tau}^{(m)}_{\boldsymbol{A}}(X,Y)=0$ for some $m\in \mathbb{N}\backslash\{0\}$, we shall say that $\bs{A}$ is a generalized Nijenhuis operator of level $m$.
\end{definition}
We recall a result which is crucial in the following analysis.
\begin{theorem} \cite{TT2022CMP} \label{th:MainR}
Let $\bs{A}: \mathfrak{X}(M)\to \mathfrak{X}(M)$ be an operator. Assume that
\beq\label{eq:integ}
\tau^{(m)}_{\bs{A}}(X,Y)=\bs{0}, \qquad X, Y \in \mathfrak{X}(M)
\eeq
for some $ m \geq 1$. Then, each eigen-distribution of $\bs{A}$, as well as each direct sum of its eigen-distributions, is integrable.
\end{theorem}

\subsection{Eigen-distributions and spectral properties of generalized Nijenhuis operators}
We shall recall now some of the spectral properties of non-semisimple  operators  on a manifold. We shall focus on the case of generalized Nijenhuis operators. Let us denote  by $Spec(\boldsymbol{A}):=\{ \lambda_1(\boldsymbol{x}),
 \lambda_2(\boldsymbol{x}), \ldots, \lambda_s(\boldsymbol{x})\}$, $\bs{x}\in M$,  the set of the  eigenvalues of an operator $\boldsymbol{A}: \mathfrak{X}(M)\to \mathfrak{X}(M)$, which will always be assumed to be pointwise distinct.  We denote by
 \begin{equation} \label{eq:DisL}
 \mathcal{D}_i(\boldsymbol{x}) = \ker \Big(\boldsymbol{A}(\boldsymbol{x})-\lambda_i(\boldsymbol{x})\boldsymbol{I}\Big)^{\rho_i}, \qquad i=1,\ldots,s
 \end{equation}
the \textit{i}-th eigen-distribution of  index $\rho_i\geq 1$, which includes all the  (generalized) eigenvectors  corresponding to the eigenvalue $\lambda_i$. In Eq. \eqref{eq:DisL}, $\rho_i$ stands for the Riesz index of $\lambda_i$, which is the minimum integer such that
\begin{equation} \label{eq:Riesz}
\ker \Big(\boldsymbol{A}(\boldsymbol{x})-\lambda_i(\boldsymbol{x})\boldsymbol{I}\Big)^{\rho_i}\equiv \ker \Big(\boldsymbol{A}(\boldsymbol{x})-\lambda_i(\boldsymbol{x})\boldsymbol{I}\Big)^{\rho_{i}+1} \ .
\end{equation}
 In the forthcoming considerations, we shall always suppose a \textit{regularity condition}, namely that the rank of the distributions and $\rho_i$ are (locally) independent of $\boldsymbol{x}$. 
When $\rho_i=1$,  $\mathcal{D}_i$ is a proper eigen-distribution.

It is also useful, from an applicative point of view, to consider the explicit action of a generalized Nijenhuis torsion on a pair of generalized eigenvectors of $\boldsymbol{A}$. In \cite{TT2022CMP} the following formula has been proved by induction over the integers $m\geq 2$:
\begin{proposition} \label{lm:TmL2autog}
Let $\boldsymbol{A}$ be a (1,1)-tensor and $X_\alpha\in\mathcal{D}_\mu$, $Y_\beta\in\mathcal{D}_\nu$ be two of its 
 generalized eigenvectors  corresponding to the eigenvalues $\mu$, $\nu$ respectively.
 Then, for any integer $m\geq 2$ the following formula holds:
\begin{equation}\label{eq:TmL2autog}
\mathcal{\tau}^{(m)}_ {\boldsymbol{A}} (X_\alpha, Y_\beta)=
\sum_{i,j=0}^{m}(-1)^{i+j}\binom{m}{i}\binom{m}{j} \Big(\boldsymbol{A}-\mu\mathbf{I}\Big)^{m-i}\Big(\boldsymbol{A}-\nu \mathbf{I}\Big)^{m-j}
 [X_{\alpha-i}, Y_{\beta-j}]   \ ,
\end{equation}
\end{proposition}
where  
\begin{equation} \label{eq:Lautg}
\boldsymbol{A} X_\alpha = \mu X_\alpha +X_{\alpha -1},\qquad\boldsymbol{A} Y_\beta = \nu Y_\beta+Y_{\beta-1} \ ,\qquad 1\leq\alpha\leq \rho_\mu,\quad 1\leq\beta\leq\rho_\nu \ ,
\end{equation}
and $X_0$ and $Y_0$  are, by definition,  null vector fields.
\subsection{Block-diagonalization}
As a nontrivial application of Theorem \ref{th:MainR}, one can also prove that, given an operator $\bs{A}$,  condition 
\eqref{eq:integ} is also sufficient to ensure the existence of a local chart where the operator $\bs{A}$ can be \textit{block-diagonalized}. We envisage relevant applications, for instance, in the theory of hydrodynamic-type systems \cite{BogJMP}, in the study of partial separability of Hamiltonian systems \cite{CR2019} and, more generally, in the context of Courant's problems for
first-order hyperbolic systems of partial differential equations \cite{CH1962}.
 
Let $\bs{A}$ be an operator satisfying condition \eqref{eq:integ}; we denote by $r_i$ the rank of the distribution $\mathcal{D}_i$ of $\bs{A}$. We also introduce the  distribution (of corank $r_i$) 
\begin{equation}
 \label{eq:E}
 \mathcal{E}_i := Im \Bigl(\boldsymbol{A}-\lambda_i\mathbf{I}\Bigr)^{\rho_i }=  \bigoplus_{{j=1,\, j\neq i}} ^s  \mathcal{D}_j, \qquad\qquad i=1,\ldots,s
\end{equation}
which is  spanned by all the generalized eigenvectors of  $\boldsymbol{A}$, except those associated with the eigenvalue $\lambda_i$ (we remind that by hypothesis  $\boldsymbol{A}$  has real eigenvalues). We shall say that $ \mathcal{E}_i$ is a \emph{characteristic distribution} of  $\boldsymbol{A}$. Let $\mathcal{E}^{\circ}_{i}$ denote the annihilator of the distribution $\mathcal{E}_{i}$. 
The   cotangent spaces of $M$ can be decomposed as
\begin{equation}
 \label{eq:TMdscomp}
T_{\boldsymbol{x}}^*M=\bigoplus_{i=1} ^s  \mathcal{E}_{i}^{\circ}(\boldsymbol{x}).
\end{equation}
 As a consequence of Theorem \ref{th:MainR},   each  characteristic distribution    $\mathcal{E}_i$  is  integrable. By $ \mathrm{E}_i$ we denote the foliation associated  with $\mathcal{E}_i$  and by
$E_i(\boldsymbol{x})$ the connected leave through $\boldsymbol{x}$,  belonging to $ \mathrm{E}_i$. Given  the set  of distributions $\{\mathcal{E}_1, \mathcal{E}_2, \ldots ,\mathcal{E}_s\}$,
we have associated an equal number of foliations $\{ \mathrm{E}_1,  \mathrm{E}_2, \ldots ,  \mathrm{E}_s\}$.   This set of foliations
 is said to be the \textit{characteristic  web} of $\boldsymbol{A}$.  The leaves $E_i(\boldsymbol{x})$ of each foliation $ \mathrm{E}_i$ are usually referred to as the  \emph{characteristic fibers} of the web.
  
  \begin{definition}
Let $\boldsymbol{A}: \mathfrak{X}(M)\to \mathfrak{X}(M)$ be an operator satisfying Eq. \eqref{eq:integ}. A collection of $r_i$ smooth functions will be said to be adapted to the foliation $\mathrm{E}_i$ of the  characteristic  \textit{web} of $\bs{A}$  if the level sets of such functions coincide with the characteristic fibers of $\mathrm{E}_i$.
\end{definition}
\begin{definition}
Let $\boldsymbol{A}: \mathfrak{X}(M)\to \mathfrak{X}(M)$ be an operator satisfying Eq. \eqref{eq:integ}. A parametrization of the characteristic web of $\boldsymbol{A}$   is an ordered set  of $n$ independent smooth functions listed as
$(\boldsymbol{f}^1, \ldots,\boldsymbol{f}^i,\ldots, \boldsymbol{f}^s) $,
such that for any $i=1,\ldots, s$, the ordered subset
$\boldsymbol{f}^i=(f^{i,1}, \ldots , f^{i,r_i})$ is adapted to the  $i$-th characteristic foliation of the web:
\begin{equation}
\label{eq:fad}
f^{i,k}_{\vert   E_i(\mathbf{x})}=c^{i,k}  \qquad \forall E_i (\mathbf{x})\in  \mathrm{E}_i \ ,\quad
k=1,\ldots,r_i\ .
\end{equation}
Here $c^{i,k}$ are real constants depending  on the indices $i$ and $k$ only.
In this case, we shall say that the  collection of these functions is adapted to the web and that each of them is a \emph{characteristic function}.
\end{definition}
In the case of a single operator with a vanishing higher-order torsion, the following result gives a simple and very general tensorial criterion for the existence of local coordinates ensuring block-diagonalizability.

\begin{theorem}\cite{TT2022CMP}\label{prop:1}
Let $\bs{A}: \mathfrak{X}(M) \to \mathfrak{X}(M)$ be an operator. 
\noi If 
 \beq
 \tau^{(m)}_{\bs{A}}(X,Y)=\bs{0} \ ,\qquad  X, Y \in \mathfrak{X}(M)
 \eeq
 for some $m\geq 1$, then  $\boldsymbol{A}$ admits local charts where it takes a block-diagonal form. 

\end{theorem}

One of the main achievement of this work is to generalize this result to the case of families of commuting operators. To this aim, we shall introduce a new class of operator algebras. 

\section{Generalized Haantjes algebras}

\subsection{Definitions}
The notion of Haantjes algebra has been introduced and discussed in \cite{TT2021JGP}. In this section, we define a class of new, generalized Haantjes algebras.

\begin{definition}\label{def:HM}
A generalized Haantjes algebra of level $l$ is a pair    $(M, \mathscr{H}^{(l)})$ with the following properties:
\begin{itemize}
\item
$M$ is a differentiable manifold of dimension $\mathrm{n}$;
\item
$ \mathscr{H}^{(l)}$ is a set of Haantjes  operators $\boldsymbol{K}:\mathfrak{X}(M)\to \mathfrak{X}(M)$ whose  Haantjes torsion of level $l$ vanishes: $\tau^{(l)}_{\bs{K}}= \ob$. Also, they  generate:
\begin{itemize}
\item
a free module over the ring of smooth functions on $M$:
\begin{equation}
\label{eq:Hmod}
\mathcal{\tau}^{(l)}_{\bigl( f\boldsymbol{K}_{1} +
                             g\boldsymbol{K}_2\bigr)}(X,Y)= \mathbf{0}
 \ , \qquad\forall\, X, Y \in \mathfrak{X}(M) \ , \quad \, f,g \in C^\infty(M)\  ,\quad \forall ~\boldsymbol{K}_1,\boldsymbol{K}_2 \in  \mathscr{H}^{(l)};
\end{equation}
  \item
a ring  w.r.t. the composition operation
\begin{equation}
 \label{eq:Hring}
\mathcal{\tau}^{(l)}_{\bigl(\boldsymbol{K}_1 \, \boldsymbol{K}_2\bigr)}(X,Y)=\mathbf{0} \ , \qquad
\forall\, \boldsymbol{K}_1,\boldsymbol{K}_2\in  \mathscr{H}^{(l)} , \quad\forall\, X, Y \in \mathfrak{X}(M)\ .
\end{equation}
\end{itemize}
\end{itemize}
If
\begin{equation}
\boldsymbol{K}_1\,\boldsymbol{K}_2=\boldsymbol{K}_2\,\boldsymbol{K}_1 \ , \quad\qquad\ \boldsymbol{K}_1,\boldsymbol{K}_2 \in  \mathscr{H}^{(l)}\ ,
\end{equation}
the  algebra $(M, \mathscr{H}^{(l)})$ will be said to be an Abelian generalized Haantjes algebra. Moreover, if   the identity operator $\boldsymbol{I}\in \mathscr{H}^{(l)}$, then $(M, \mathscr{H}^{(l)})$ will be said to be a generalized Haantjes algebra with identity.
\end{definition}
In other words, the set $\mathscr{H}^{(l)}$ can be regarded as an associative algebra of generalized Haantjes operators. 


The case $l=2$, namely that of Haantjes algebras \cite{TT2021JGP}, possesses several important properties. Among them, we recall that for a given Abelian Haantjes algebra $\mathscr{H}^{(2)}\equiv\mathscr{H}$ there exists  an associated set of  coordinates, called \textit{Haantjes coordinates},   by means of which  all $\boldsymbol{K}\in \mathscr{H}$ can  be  written simultaneously in a  block-diagonal form.  
In particular, if $\mathscr{H}$ is also semisimple, on each set of Haantjes coordinates all $\boldsymbol{K}\in \mathscr{H}$ can  be written simultaneously  in a purely diagonal form \cite{TT2021JGP}. We mention that Haantjes algebras play a relevant role in the theory of classical  separable and multiseparable Hamiltonian systems \cite{RTT2022CNS}.

\vspace{2mm}

The following result simplifies the study of the integrability of the eigen-distributions of a family of operators forming an Abelian, generalized Haantjes algebra.
 
 \begin{lemma}\label{th:HaanA}
Let $(M, \mathscr{H}^{(l)})$ be an Abelian generalized Haantjes algebra of level $l$. We shall assume that the rank of the  eigen-distributions of the operators belonging to $\mathscr{H}^{(l)}$ is independent of $\boldsymbol{x}\in M$.
Then, each nontrivial intersection of these   eigen-distributions is integrable.
\end{lemma}
\begin{proof}
Let $\{\boldsymbol{K}_1,\ldots,\boldsymbol{K}_w\}$ be a basis of $(M, \mathscr{H}^{(l)})$, and
$ (\mathcal{D}_{i_1}^{(1)},  \ldots    ,\mathcal{D}_{i_w}^{(w)})$, $i_1=1,\ldots, s_1$, $i_w=1,\ldots,s_w$, be the set of their (proper) eigen-distributions. Let
\begin{equation}
\label{eq:1ints}
 \mathcal{V}_a=\mathcal{D}_{i_1}^{(1)}(\boldsymbol{x}) \bigcap  \ldots      \bigcap \mathcal{D}_{i_w}^{(w)}(\boldsymbol{x})\qquad a=1,\ldots, v, \qquad v\leq n
\end{equation}
denote a nontrivial intersection of eigen-distributions of the operators $\{\boldsymbol{K}_1,\ldots,\boldsymbol{K}_w\}$.  Consequently,  this distribution,  being the  intersection  of distributions which are involutive  due to Theorem \ref{th:MainR}  \cite{TT2022CMP}, is also involutive.

\end{proof}


\subsection{Cyclic generalized Haantjes algebras}
A special class of commutative generalized Haantjes algebras are those generated by a single operator and its independent powers.

\begin{definition}
A generalized Haantjes algebra $(M, \mathscr{H}^{(l)})$ will be said to be cyclic if there exists an operator $\bs{L}:\mathfrak{X}(M)\rightarrow \mathfrak{X}(M)$, with $\tau_{\bs{L}}^{(l)}(X,Y)=0$,  such that $\mathscr{H}^{(l)}\subseteq\mathcal{L}(\bs{L})$, where $\mathcal{L}(\bs{L})= \langle \bs{I}, \bs{L}, \ldots, \bs{L}^{k},\ldots \rangle $ is the algebra generated by all the independent powers of $\bs{L}$.

\end{definition}

Given a cyclic algebra $(M, \mathscr{H}^{(l)})$ generated by an operator $\bs{L}$, its rank is less than or equal to the degree of the minimal polynomial of $\bs{L}$.

\vspace{2mm}

\begin{remark}
Cyclic generalized Haantjes algebras as algebraic structures are not rare. Indeed, as we will prove, given an operator $\bs{L}$ with $\tau_{\bs{L}}^{(l)}(X,Y)=0$,  any polynomial in $\bs{L}$ having coefficients in $C^{\infty}(M)$ is still an operator having the same  level-$l$ vanishing torsion. Thus, the independent powers of any operator with a given vanishing generalized torsion provide naturally a basis of a cyclic generalized algebra of the same level.
\end{remark}


In the following, we shall prove this fundamental fact. To this aim, we shall use and adapt to our general case several algebraic techniques and results introduced and proved by Bogoyavlenskij in \cite{BogIzv2004} for the case of Nijenhuis and Haantjes operators. In particular, we shall use a useful representation of (1,2)-tensors in the ring $S_3$ of polynomials of three independent variables $z,\lambda,\mu$  with coefficients depending on $\bs{x}\in M$. A generic polynomial of this ring has the form
\[
S(z,\lambda,\mu):= \sum_{i,j,k}^{N} s_{ijk}(\bs{x})z^{i}\lambda^{j}\mu^{k}, \qquad \bs{x}\in M
\]
with $N\in \mathbb{N}\backslash\{0\}$, $s_{ijk}(\bs{x})\in C^{\infty}(M)$. Given an operator $\bs{A}$ and a (1,2)-tensor $T(X,Y)$, we introduce the representation defined by \cite{BogIzv2004}
\beq \label{eq:rep}
R_S(T)(X,Y) = \sum_{i,j,k}^{N} s_{ijk}(\bs{x})\bs{A}^{i}T(\bs{A}^{j}X, \bs{A}^k Y) \ .
\eeq
Thus, the action of $\lambda$ and $\mu$ is associated with the first and second arguments of $T(X,Y)$, whereas the action of $z$ is associated with the valued of $T(X,Y)$. Representation \eqref{eq:rep} satisfies the basic properties
\beq \label{eq:proprep1}
R_{S_1+S_2}=R_{S_1}+R_{S_2},
\eeq
\beq \label{eq:proprep2}
 \qquad R_{S_1 \cdot S_2}=R_{S_2 \cdot S_1}= R_{S_1}\cdot R_{S_2} \ .
\eeq
This framework allows us to represent the ``tower" of generalized Nijenhuis torsions in a direct way.
\begin{lemma} \label{lemma2}
Let $\bs{A}: \mathfrak{X}(M)\rightarrow \mathfrak{X}(M)$ be an operator. We have:
\beq
\mathcal{\tau}^{(m+1)}_{\boldsymbol{A}}(X,Y)=R_\sigma \tau^{(m)}_{\boldsymbol{A}}(X,Y)
\eeq
where $\sigma$ is the polynomial $\sigma(z,\lambda,\mu)=(z-\lambda)(z-\mu)$.
\end{lemma}
\begin{proof}
It is an immediate consequence of the formula \eqref{GNTn}, which indeed can be written as
\bea
\nn R_{(z^2-z\lambda-z\mu+\lambda\mu)} \mathcal{\tau}^{(m)}_{\boldsymbol{A}}(X,Y)&=&\boldsymbol{A}^2\mathcal{\tau}^{(m)}_{\boldsymbol{A}}(X,Y)+
\mathcal{\tau}^{(m)}_{\boldsymbol{A}}(\boldsymbol{A}X,\boldsymbol{A}Y) \\ \nn  &-& 
\boldsymbol{A}\Big(\mathcal{\tau}^{(m)}_{\boldsymbol{A}}(X,\boldsymbol{A}Y)+\mathcal{\tau}^{(m)}_{\boldsymbol{A}}(\boldsymbol{A}X,Y)\Big) \\ &=&\mathcal{\tau}^{(m+1)}_{\boldsymbol{A}}(X,Y) \ .
\eea
\end{proof}
When $m=2$, we recover the formula 
\[
\mathcal{H}_{\bs{A}}(X,Y)= R_{\sigma}\tau_{\bs{A}}(X,Y)
\]
first stated in \cite{BogIzv2004}. We also remind that, given a polynomial $P(z)=\sum_{k=0}^{N}c_k(\bs{x})z^{k}$, the \textit{B\'ezout identity} holds:
\beq \label{eq:BI}
P(z)-P(\lambda)=(z-\lambda)Q_{P}(z,\lambda)
\eeq
where 
\[
Q_{P}(z,\lambda)=\sum_{k=1}^{N}c_k(\bs{x})\sum_{p+q=k-1}z^{p}\lambda^{q} \ .
\]

We prove now a useful technical result.
\begin{lemma} \label{lemma3}
Let $\bs{A}:\mathfrak{X}(M)\rightarrow \mathfrak{X}(M)$ be an operator. Let  $\bs{P}:=P(\bs{A})= \sum_{k=0}^{N} c_{k}(\bs{x})\bs{A}^{k}$ be a polynomial in $\bs{A}$ with variable coefficients. We have
\beq \label{lemma2eq0}
\mathcal{\tau}_{\bs{P}}^{(m)}(X,Y)=R_{(Q_{P}(z,\lambda))^m(Q_{P}(z,\mu))^m}\mathcal{\tau}_{\bs{A}}^{(m)}(X,Y), \qquad m\geq 2 \ .
\eeq
\end{lemma}
\begin{proof}
We shall proceed by induction over $m\geq 2$. The case $m=2$, corresponding to the Haantjes torsion,  has been proved in \cite{BogIzv2004}. Thus, we assume that the property is true for the case of a torsion of level $m-1$ and we prove that it holds true for the case of a torsion of level $m$. According to Lemma \ref{lemma2}, we have
\beq\label{lemma2eq2}
 \mathcal{\tau}^{(m)}_{\boldsymbol{A}}(X,Y)=R_{\sigma} \mathcal{\tau}^{(m-1)}_{\boldsymbol{A}}(X,Y)
\eeq
where  $\sigma=(z-\lambda)(z-\mu)$. Also, in terms of the operator $\bs{P}$ expression \eqref{lemma2eq2} can be  written as
\beq
 \mathcal{\tau}^{(m)}_{\bs{P}}(X,Y)=R_{(P(z)-P(\lambda))(P(z)-P(\mu))}\mathcal{\tau}^{(m-1)}_{\bs{P}}(X,Y) \ .
\eeq
Thus, applying twice the Bezout identity \eqref{eq:BI}, we get
\[
(P(z)-P(\lambda))(P(z)-P(\mu))=(z-\lambda)(z-\mu)Q_P(z,\lambda)Q_P(z,\mu) \ .
\]
Exploiting property \eqref{eq:proprep2}, we obtain
\beq \label{eq:4.5}
\mathcal{\tau}^{(m)}_{\bs{P}}(X,Y)=R_{Q_P(z,\lambda)Q_P(z,\mu)}R_{\sigma}\mathcal{\tau}^{(m-1)}_{\bs{P}}(X,Y) \ . 
\eeq
By the induction hypothesis, we have
\[
\mathcal{\tau}^{(m-1)}_{\bs{P}}(X,Y)= R_{(Q_{P}(z,\lambda))^{m-1}(Q_{P}(z,\mu))^{m-1}}\mathcal{\tau}_{\bs{A}}^{(m-1)}(X,Y) \ .
\]
Consequently, due to eq. \eqref{eq:4.5} and Lemma \ref{lemma2}, we deduce that
\beq
\mathcal{\tau}^{(m)}_{\bs{P}}(X,Y)=R_{(Q_P(z,\lambda))^{m}(Q_P(z,\mu))^{m}}R_\sigma \tau^{(m-1)}_{\bs{A}}(X,Y)= R_{(Q_P(z,\lambda))^{m}(Q_P(z,\mu))^{m}}\tau^{(m)}_{\bs{A}}(X,Y) \ ,
\eeq
which completes the proof.
\end{proof}
We can now state the main result of this section.

\begin{theorem} \label{th:2a}
Let $\bs{A}:\mathfrak{X}(M)\rightarrow \mathfrak{X}(M)$ be an operator with $\tau_{\bs{A}}^{(m)}(X,Y)=0$, $m\geq 2$. Then 
\beq \label{eq:conj}
\tau^{(m)}_{\big(\sum_{k=0}^{N} c_{k}(\bs{x})\bs{A}^{k}\big)}(X,Y) =0, \qquad X, Y\in \mathfrak{X}(M) \ , 
\eeq
where $c_{k}(\bs{x})\in C^{\infty}(M)$.
\end{theorem}
\begin{proof}
Let $\bs{P}:=P(\bs{A})= \sum_{k=0}^{N} c_{k}(\bs{x})\bs{A}^{k}$. It is sufficient to apply Lemma \ref{lemma3}, taking into account that  $\tau_{\bs{A}}^{(m)}(X,Y)=0$. 
\end{proof}

Interestingly enough, Theorem \ref{th:2a} does not hold for the case of the Nijenhuis torsion (i.e., $m=1$). As is well-known, the relation \eqref{eq:conj} is valid for the Nijenhuis torsion only if the coefficients $c_k(\bs{x})$ are all constant. 
\section{Generalized Nijenhuis torsions and simultaneous block-diagonalization}

\begin{theorem}\label{th:1}
Let  $\mathcal{S}=\{\boldsymbol{K}_1,\ldots,\boldsymbol{K}_w\}$, $\boldsymbol{K}_\alpha: \mathfrak{X}(M)\rightarrow \mathfrak{X}(M)$, $\alpha=1,\ldots,w$ be  a set  of commuting   operator fields; we assume that  one of them, say $\boldsymbol{K}_1$, satisfies the condition 
\beq \label{eq:44-1}
\tau^{(l)}_{\bs{K}_1}(X,Y)=\bs{0}, \qquad X, Y \in \mathfrak{X}(M) \ ,
\eeq
for some $l\geq 1$.  Then, there exist local charts in which all of the operators $\boldsymbol{K}_\alpha$ can be written simultaneously in a block-diagonal form.

\end{theorem}
\begin{proof}
Assuming that  condition \eqref{eq:44-1} is satisfied, Proposition \ref{prop:1} ensures the existence of an equivalence class of integrable frames and local charts where the  operator $\boldsymbol{K}_1$ takes a block-diagonal form.
Such coordinates  are adapted to the characteristic web  associated with the spectral decomposition of $\boldsymbol{K}_1$:
\begin{equation}
\label{eq:decomp}
T_{\boldsymbol{x}} M= \mathcal{D}_{i}(\boldsymbol{x})\bigoplus \mathcal{E}_{i}(x)=\bigoplus_{i=1}^{s}\mathcal{D}_{i}(\boldsymbol{x}) \ .
\end{equation}
with
\begin{equation}
\label{eq:K1chart}
\boldsymbol{x}= (\boldsymbol{x}^1,\ldots,\boldsymbol{x}^{i}, \ldots,\boldsymbol{x}^{s}) \ .
\end{equation}
The variables $\boldsymbol{x}^{i}=(x^{i,1},\ldots,  x^{i,r_{i}})$ are defined over the integral leaves
 of the eigen-distribution $\mathcal{D}_{i}$, whereas the remaining ones, namely
\[
(\boldsymbol{x}^1,\ldots, \boldsymbol{x}^{i-1}, \boldsymbol{x}^{i+1},\ldots,\boldsymbol{x}^{s})
\]
 are coordinates of the leaves, i.e.,~are constant
 $(\boldsymbol{x}^1=\boldsymbol{c}^1,\ldots, \boldsymbol{x}^{i-1}=\boldsymbol{c}^{i-1}, \boldsymbol{x}^{i+1}=\boldsymbol{c}^{i+1}, \ldots,\boldsymbol{x}^{s}=\boldsymbol{c}^{s})$
 on each leaf $D_{i}(\boldsymbol{c})$ of the foliation. 
Here
\noi  $\boldsymbol{c}:=(\boldsymbol{c}^1,\ldots, \boldsymbol{c}^{i-1}, \boldsymbol{c}^{i+1}, \ldots,\boldsymbol{c}^{s}) $.
Since all operators of the set $\mathcal{F}$ commute, every distribution $\mathcal{D}_{i}$ is invariant under the action of the operators $\{\boldsymbol{K}_2,\ldots,\boldsymbol{K}_w\}$. As a direct consequence of this property, all the operators  $\boldsymbol{K}_\alpha\in \mathcal{F}$ in the local chart \eqref{eq:K1chart}  take a block-diagonal form, where the $i$-th ($r_i \times r_i$)   block matrix $[\boldsymbol{K}_{\alpha}^{(i)}]_{ \vert _{D_{i}(\boldsymbol{c})}}$ coincides with the matrix $[\boldsymbol{K}_{\alpha \vert _{D_{i}(\boldsymbol{c})}}]$, $i=1,\ldots,s$, which is the representation of the restricted operator $\bs{K}_\alpha$ to the leaf $D_{i}(\boldsymbol{c})$.
\end{proof}

\begin{theorem}\label{th:2}
Let  $\mathcal{S}=\{\boldsymbol{K}_1,\ldots,\boldsymbol{K}_w\}$, $\boldsymbol{K}_\alpha: \mathfrak{X}(M)\rightarrow \mathfrak{X}(M)$,  be  a set  of commuting   operator fields. If

(i) all the operators of the family have  vanishing generalized torsion of level $l$:
\beq \label{eq:5.5}
\tau^{(l)}_{\bs{K}_\alpha}(X,Y)=\bs{0} \ , \qquad \alpha=1,\ldots,q , \qquad X, Y \in \mathfrak{X}(M)\rightarrow \mathfrak{X}(M)
\eeq

(ii) all possible nontrivial intersections of their generalized eigen-distributions
\begin{equation}
\label{eq:Va}
\mathcal{V}_a(\boldsymbol{x}):=  \bigoplus_{i_1,\ldots,i_w}^{s_1,\ldots,s_w}\mathcal{D}_{i_1}^{(1)}(\boldsymbol{x}) \bigcap  \ldots      \bigcap \mathcal{D}_{i_w}^{(w)}(\boldsymbol{x}) \ ,\qquad a=1,\ldots, v\leq n
\end{equation}
are mutually integrable,
 then  there exist  sets of local coordinates, adapted to the decomposition
\begin{equation}
 \label{eq:TVa}
T_{\boldsymbol{x}}M= \bigoplus_{a=1}^{v}\mathcal{V}_a( \boldsymbol{x}) \qquad  \boldsymbol{x} \in M ,
\end{equation}
in which all operators $\boldsymbol{K}_\alpha$  admit  simultaneously a block-diagonal form with  possibly finer blocks.
\end{theorem}

\begin{proof} According to Theorem \ref{th:1}, since by hypothesis $\tau^{(l)}_{\bs{K}_1}(X,Y)=\bs{0}$, $X, Y \in \mathfrak{X}(M)$, all the operators take simultaneously a block-diagonal form in a suitable local chart.  Assume also that $\tau^{(l)}_{\bs{K}_2}(X,Y)=\bs{0}$. Then, the tangent space at any point $\boldsymbol{x} $ admits the finer decomposition
\begin{equation}
 \label{eq:TMinters12}
T_{\boldsymbol{x}}M= \bigoplus_{i_1,i_2}^{s_1,s_2}\mathcal{D}_{i_1}^{(1)}(\boldsymbol{x}) \cap \mathcal{D}_{i_2}^{(2)}(\boldsymbol{x}) \ ,
\end{equation}
where $\mathcal{D}_{i_2}^{(2)}$ are  the eigen-distributions    of $\boldsymbol{K}_2$. These eigen-distributions are integrable
by virtue of Theorem \ref{th:MainR}.
Consequently, the Haantjes Theorem can also be applied to the restriction of $\boldsymbol{K}_2$ to  $D^{(1)}_{i_1}(\boldsymbol{c}_1)$. Therefore, there exists a transformation of coordinates, acting only on the coordinates over the leaves of the foliation $D^{(1)}_{i_1}$
\begin{equation}
\Phi: M\rightarrow M, \qquad
(\boldsymbol{x}^1,\ldots,\boldsymbol{x}^{i_1}, \ldots,\boldsymbol{x}^{s_1}) \mapsto
(\boldsymbol{x}^1,\ldots,\boldsymbol{y}^{i_1}, \ldots,\boldsymbol{x}^{s_1}) \ ,
\end{equation}
such that the new coordinates $\boldsymbol{y}^{i_1}= (y^{i_1, 1},\ldots,y^{i_1 ,r_{i_1}})=\boldsymbol{f}^{i_1}(\boldsymbol{x}^{i_1})$ are
adapted to the decomposition
\begin{equation}
 \label{eq:TMinters12}
T_{\boldsymbol{x}}D^{(1)}_{i_1}(\boldsymbol{c}_1)= \bigoplus_{i_2}^{s_2}\mathcal{D}_{i_1}^{(1)}(\boldsymbol{x}) \cap \mathcal{D}_{i_2}^{(2)}(\boldsymbol{x}), \qquad \boldsymbol{x} \in
D^{(1)}_{i_1}(\boldsymbol{c}_1)\ .
\end{equation}
 Thus, we have
    \begin{equation}
\label{eq:K2diag}
 [\boldsymbol{K}^{(i_1)}_{\alpha}]=
\left[\begin{array}{c|c|c}
\boldsymbol{K}_{\alpha}^{(i_1,1)} & 0 & 0 \\
\hline 0 & \ddots & 0 \\
\hline 0 & 0 & \boldsymbol{K}_{\alpha}^{(i_1,s_2)}
\end{array}\right] \ , \qquad \alpha=1, \dots, w
\end{equation}
 where $[\boldsymbol{K}_{\alpha}^{(i_1,j)}]_
 { \vert _{D_{i_1}^{(1)}(\boldsymbol{c}_1) \cap D_{j}^{(2)}(\boldsymbol{c}_2)}} =\left[\boldsymbol{K}_
 {\alpha \vert _{D_{i_1}^{(1)}(\boldsymbol{c}_1) \cap D_{j}^{(2)}(\boldsymbol{c}_2)}} \right]$,
 $j=1,\ldots,s_2$.
  Let us consider  the decomposition
\begin{equation}
 \label{eq:TMinters12}
T_{\boldsymbol{x}}M= \bigoplus_{i_1,i_2}^{s_1,s_2}\mathcal{D}_{i_1}^{(1)}(\boldsymbol{x}) \cap \mathcal{D}_{i_2}^{(2)}(\boldsymbol{x})=
\bigoplus_{\gamma=1}^{u}\mathcal{U}_\gamma (\boldsymbol{x})\qquad \boldsymbol{x} \in M\ ,
\end{equation}
where in the direct sum \eqref{eq:TMinters12}  $\mathcal{U}_\gamma \neq\{\boldsymbol{0}\}$, $u\leq n$ and $r_\gamma$ denotes the rank of  $\mathcal{U}_\gamma$ ($\sum_{\gamma=1}^u r_\gamma=n$). Clearly, the  distributions  $\mathcal{U}_\gamma$ are invariant under the action of each operator $\boldsymbol{K}_\alpha\in\mathcal{S}$. Besides, these distributions are involutive and are realized as the intersection of involutive distributions. By assumption, they are also \emph{mutually} integrable. Therefore, there exist  local charts in $M$ of the form
\begin{equation}
\label{eq:diag12}
\{ U ,(y^{1,1},\ldots,y^{1,r_{1}}; \dots;y^{i_1, 1},\ldots,y^{i_1 ,r_{i_1}};\ldots ;y^{s_{1},1},\dots,y^{s_{1},r_{s_1}})\}
\end{equation}
adapted to the decomposition \eqref{eq:TMinters12}, where all the operators $\boldsymbol{K}_\alpha\in\mathcal{S}$  admit simultaneously a (possibly) finer block-diagonal form.
By extending the previous procedure to the Haantjes operators $\boldsymbol{K}_3, \ldots, \boldsymbol{K}_w$, we obtain the decomposition
\begin{equation}
 \label{eq:TMinters}
T_{\boldsymbol{x}}M= \bigoplus_{i_1,\ldots,i_w}^{s_1,\ldots,s_w}\mathcal{D}_{i_1}^{(1)}(\boldsymbol{x}) \bigcap  \ldots      \bigcap \mathcal{D}_{i_w}^{(w)}(\boldsymbol{x})= \bigoplus_{a=1}^{v}\mathcal{V}_a(  \boldsymbol{x})  ,
\end{equation}
where in the  direct sum \eqref{eq:TMinters} $\mathcal{V}_a \neq \{\boldsymbol{0}\}$, $v\leq n$ and $r_a$ denotes the rank of $\mathcal{V}_a$  ($\sum_{a=1}^v r_a=n$).
Then, as the involutive distributions $\mathcal{V}_a$ by assumption are also mutually integrable,  there exist  local charts
\begin{equation}
\label{eq:HchartK}
\{ U,(\boldsymbol{y}^{1},\ldots , \boldsymbol{y}^{a},\ldots,   \boldsymbol{y}^{v}) \} 
\end{equation}
adapted to the decomposition \eqref{eq:TMinters}, such that
\begin{equation}
 \label{eq:Va}
\mathcal{V}_a=\left \langle  \frac{\partial}{\partial y^{a, 1}},\ldots , \frac{\partial}{\partial y^{a, r_a}} \right  \rangle  \qquad a=1,\ldots,v,
\end{equation}
where the natural frame $\big\{ \frac{\partial}{\partial y^{a, 1}},\ldots , \frac{\partial}{\partial y^{a, r_a}} \big\}$ over the leaves of $\mathcal{V}_a$ is formed by joint \emph{generalized} eigenvector fields of the operators $\{\boldsymbol{K}_1,\ldots,\boldsymbol{K}_w\}$.
\end{proof}

\begin{remark}
Given a cyclic generalized Haantjes algebra $(M, \mathscr{H}^{(l)})$ with generator $\bs{L}$, the hypothesis $(ii)$ of Theorem \ref{th:2} is automatically satisfied. Indeed, the intersections $\mathcal{V}_a$ of eigen-distributions of the operators forming the algebra $\mathscr{H}^{(l)}$ either coincide with the eigen-distributions of the 
generator $\bs{L}$ or are direct sums of them. In both cases, since (due to  Theorem \ref{th:MainR}) the eigen-distributions of $\bs{L}$ are mutually integrable, it turns out that the direct sums of the intersections $\mathcal{V}_a$ are integrable as well.

\end{remark}

In the following sections, we shall illustrate the results proposed. First, we will show some examples of generalized Haantjes algebras of operator fields. Then, we shall present the explicit construction of local charts of coordinates allowing their simultaneous block-diagonalization.

\section{Block-diagonalization of generalized Haantjes algebras}

As stated in Theorems \ref{th:1} and \ref{th:2}, we can block-diagonalize families of commuting generalized Nijenhuis operators by constructing a suitable coordinate chart. To this aim, we propose the following procedure.

\vspace{2mm}

1) Determine the joint eigen-distributions of the given family of operator fields \\
2) Determine a basis of one-forms for each of the the corresponding annihilators \\
3) Integrate them to find the characteristic coordinates \\
4) Compute the expression of the operators of the given family in these coordinates

\vspace{2mm}

To illustrate this procedure, we shall first construct families of operator fields, depending on arbitrary functions, having a prescribed vanishing higher-level torsion. A direct approach consists in considering an operator field whose entries are all arbitrary functions, and imposing both that a certain higher-level torsion is vanishing, and simultaneously that the lower-level ones are not.
In this way, one can obtain a set of differential constraints that can be solved (possibly with suitable ans\"atze) to obtain families of operator fields with a vanishing higher-level torsion, still depending on arbitrary functions. Generally speaking, there is a lot of freedom in this approach.  

From this analysis it emerges that generalized Haantjes algebras are not rare. Thus, by specializing the arbitrary functions contained in the families of operators determined according to the procedure proposed, we can obtain infinitely many examples of new, higher-level Haantjes algebras.

\subsection{A level-three generalized Haantjes algebra}

Let $M$ be a five-dimensional differentiable manifold. A first example is given, in local coordinates $\bs{x}=(x^{1}, x^{2}, x^{3}, x^{4},x^{5})$ by the family of operators
\begin{equation}
\bs{L}^{(\bs{x})} = \label{eq:LTA}
\left[\begin{array}{ccccc}
f_{1} & 1 & 0 & 1 & 0 \\
f_{1} - f_{2} + 1 & f_{1} + 1 & - f_{3} & f_{1} - f_{2} + 1 & - f_{3} \\
1 & 0  & f_{2} + f_{3} & 1 & f_{3} \\
f_{2} - f_{1} & -1  & f_{3} & f_{2} - 1 & f_{3} \\
- 1 & 0 & - f_{3} - 1 & - 1 & f_{2} - f_{3} - 1 \\
\end{array}\right],
\end{equation}
where $f_{1} = f_{1} (x^{1}, x^{2})$, $f_{2} = f_{2} (x^{3}, x^{4}, x^{5})$ and $f_{3} = f_{3} (x^{3})$ are arbitrary functions. 

A direct computation allows us to prove the following result.
\begin{proposition}
The family of operators $\bs{L}$ satisfies the following properties:
\begin{itemize}
\item  $\tau^{(3)}_{\bs{L}} (X,Y) = 0$ 
\item $\mathcal{H}_{\bs{A}}(X,Y)\neq 0$ if and only if $f_{3}' \neq 0$.
\end{itemize}
\end{proposition}


The minimal polynomial of $\bs{L}$ is of fifth degree  and reads
\begin{equation}
m_{\boldsymbol{L}} (\boldsymbol{x}, \lambda) = (\lambda-f_1-1)(\lambda-f_1+1)\left( \lambda - f_{2} + 1 \right)(\lambda - f_{2})^{2}  \ .
\end{equation}
Thus, it coincides with the characteristic polynomial; therefore $\bs{L}$ is a cyclic operator. Also, notice that $\boldsymbol{L}$ is generically non-semisimple. The eigenvalues and generalized eigenvectors of $\bs{L}$ are
\begin{equation}
\begin{aligned}
\lambda_{1} = f_{1} + 1, \quad \rho_{1} = 1 & \Longrightarrow \cD_{1} = \bigg\langle \dfrac{\partial}{\partial x^{1}} + 2 \dfrac{\partial}{\partial x^{2}} - \dfrac{\partial}{\partial x^{4}} \bigg\rangle, \\
\lambda_{2} = f_{1} - 1, \quad \rho_{2} = 1 & \Longrightarrow \cD_{2} = \bigg\langle \dfrac{\partial}{\partial x^{1}} - \dfrac{\partial}{\partial x^{4}} \bigg\rangle, \\
\lambda_{3} = f_{2} - 1, \quad \rho_{3} = 1 & \Longrightarrow \cD_{3} = \bigg\langle f_{3} \dfrac{\partial}{\partial x^{2}} - f_{3} \dfrac{\partial}{\partial x^{4}} + \dfrac{\partial}{\partial x^{5}} \bigg\rangle, \\
\lambda_{4} = f_{2}, \quad \rho_{4} = 2 & \Longrightarrow \cD_{4} = \bigg\langle \dfrac{\partial}{\partial x^{3}} - \dfrac{\partial}{\partial x^{5}}, \dfrac{\partial}{\partial x^{2}} - \dfrac{\partial}{\partial x^{4}}\bigg\rangle.
\end{aligned}
\end{equation}

From these eigen-distributions, we can compute  their characteristic distributions and consequently their annihilators:
\begin{equation} \label{eq:ann5}
\begin{aligned}
& \mathcal{E}_{1} = \oplus_{i \neq 1} \cD_{i} = \bigg\langle \dfrac{\partial}{\partial x^{1}} - \dfrac{\partial}{\partial x^{4}}, f_{3} \dfrac{\partial}{\partial x^{2}} - f_{3} \dfrac{\partial}{\partial x^{4}} + \dfrac{\partial}{\partial x^{5}}, \dfrac{\partial}{\partial x^{2}} - \dfrac{\partial}{\partial x^{4}}, \dfrac{\partial}{\partial x^{3}} - \dfrac{\partial}{\partial x^{5}} \bigg\rangle \\
& \hspace{2cm} \Longrightarrow \mathcal{E}_{1}^{\circ} = \langle dx^{1} + dx^{2} + dx^{4} \rangle \\
& \mathcal{E}_{2} = \oplus_{i \neq 2} \cD_{i} = \bigg\langle \dfrac{\partial}{\partial x^{1}} + 2 \dfrac{\partial}{\partial x^{2}} - \dfrac{\partial}{\partial x^{4}}, f_{3} \dfrac{\partial}{\partial x^{2}} - f_{3} \dfrac{\partial}{\partial x^{4}} + \dfrac{\partial}{\partial x^{5}}, \dfrac{\partial}{\partial x^{2}} - \dfrac{\partial}{\partial x^{4}}, \dfrac{\partial}{\partial x^{3}} - \dfrac{\partial}{\partial x^{5}} \bigg\rangle \\
& \hspace{2cm} \Longrightarrow \mathcal{E}_{2}^{\circ} = \langle dx^{1} - dx^{2} - dx^{4} \rangle \\
& \mathcal{E}_{3} = \oplus_{i \neq 3} \cD_{i} = \bigg\langle \dfrac{\partial}{\partial x^{1}} + 2 \dfrac{\partial}{\partial x^{2}} - \dfrac{\partial}{\partial x^{4}}, \dfrac{\partial}{\partial x^{1}} - \dfrac{\partial}{\partial x^{4}}, \dfrac{\partial}{\partial x^{2}} - \dfrac{\partial}{\partial x^{4}}, \dfrac{\partial}{\partial x^{3}} - \dfrac{\partial}{\partial x^{5}} \bigg\rangle \\
& \hspace{2cm} \Longrightarrow \mathcal{E}_{3}^{\circ} = \langle dx^{3} + dx^{5} \rangle, \\
& \mathcal{E}_{4} = \oplus_{i \neq 4} \cD_{i} = \bigg\langle \dfrac{\partial}{\partial x^{1}} + 2 \dfrac{\partial}{\partial x^{2}} - \dfrac{\partial}{\partial x^{4}}, \dfrac{\partial}{\partial x^{1}} - \dfrac{\partial}{\partial x^{4}}, f_{3} \dfrac{\partial}{\partial x^{2}} - f_{3} \dfrac{\partial}{\partial x^{4}} + \dfrac{\partial}{\partial x^{5}} \bigg\rangle \\
& \hspace{2cm} \Longrightarrow \mathcal{E}_{4}^{\circ} = \langle dx^{3}, dx^{1} + dx^{4} + f_{3} dx^{5} \rangle.
\end{aligned}
\end{equation}

\begin{remark} We can easily extract an Abelian cyclic algebra $\mathscr{H}^{(3)}$ from the class \eqref{eq:LTA}. In order to achieve that, we set $f_{3} = x^{3}$ for all the operators of the algebra and vary the other arbitrary functions. Observe that  all the operators of the family obtained in this way will share both the same eigen-distributions and the degree of the minimal polynomial. Thus, any of them generate the full algebra $\mathscr{H}^{(3)}$  with its independent powers. Notice that for this algebra the hypotheses of Theorem \ref{th:2} are all satisfied. 
\end{remark}
By way of an example of simultaneous diagonalization, we can consider the three commuting operators
\begin{equation}
\bs{L}_{1}^{(\bs{x})} =
\left[\begin{array}{ccccc}
x^{1} & 1 & 0 & 1 & 0 \\
x^{1} + 1 & x^{1} + 1 & - x^{3} & x^{1} + 1 & - x^{3} \\
1 & 0  & x^{3} & 1 & x^{3} \\
- x^{1} & -1  & x^{3} & - 1 & x^{3} \\
- 1 & 0 & - x^{3} - 1 & - 1 & - x^{3} - 1 \\
\end{array}\right],
\end{equation}
\begin{equation}
\bs{L}_{2}^{(\bs{x})} =
\left[\begin{array}{ccccc}
x^{2} & 1 & 0 & 1 & 0 \\
x^{2} - x^{5} + 1 & x^{2} + 1 & - x^{3} & x^{2} - x^{5} + 1 & - x^{3} \\
1 & 0  & x^{3} + x^{5} & 1 & x^{3} \\
- x^{2} + x^{5} & -1  & x^{3} & x^{5} - 1 & x^{3} \\
- 1 & 0 & - x^{3} - 1 & - 1 & - x^{3} + x^{5} - 1 \\
\end{array}\right] \ , 
\end{equation}
\begin{equation}
\bs{L}_{3}^{(\bs{x})} =
\left[\begin{array}{ccccc}
0 & 1 & 0 & 1 & 0 \\
- x^{4} - x^{3} x^{5} + 1 & 1 & - x^{3} & - x^{4} - x^{3} x^{5} + 1 & - x^{3} \\
1 & 0  & x^{4} + x^{3} \left( x^{5} + 1 \right) & 1 & x^{3} \\
x^{4} + x^{3} x^{5} & -1  & x^{3} & x^{4} + x^{3} x^{5} - 1 & x^{3} \\
- 1 & 0 & - x^{3} - 1 & - 1 & x^{4} + x^{3} \left( x^{5} - 1 \right) - 1 \\
\end{array}\right] \ .
\end{equation}

These operators are $C^{\infty}(\mathbb{R})$-linearly independent. Any of these three operators can be chosen to be the generator of a level-three Haantjes cyclic algebra $\mathscr{H}^{(3)}$ of rank $5$. A basis of this algebra is, for instance, $\mathscr{B}=\{\bs{I}, \bs{L}_{1}, \bs{L}_{1}^2,  \bs{L}_{1}^{3}, \bs{L}_{1}^{4} \}$. In the following subsection, we shall determine a set of coordinates which block-diagonalize the full algebra $\mathscr{H}^{(3)}$.  
 

\subsection{Block-diagonalization}
By integrating the annihilators \eqref{eq:ann5} with the choice $f_3=x^3$ of the  eigen-distributions of these operators, we obtain a set of coordinates which block-diagonalized the full  algebra $\mathscr{H}^{(3)}$. We have explicitly:
\begin{equation}
\begin{aligned}
& y^{1} = x^{1} + x^{2} + x^{4}, \\
& y^{2} = x^{1} - x^{2} - x^{4}, \\
& y^{3} = x^{3} + x^{5}, \\
& y^{4} = x^{3}, \\
& y^{5} = x^{1} + x^{4} + x^{3} x^{5}.
\end{aligned}
\end{equation}
In particular, in the block-separation coordinates, the operators $\bs{L}_{1}, \bs{L}_{2}, \bs{L}_{3}$ read:
\begin{equation}
\bs{L}_{1}^{(\bs{y})} =
\left[\begin{array}{ccccc}
\frac{1}{2} \left(y^{1} + y^{2} \right) + 1 & \tempr & 0 & 0 & 0 \\ \cline{1-2}
\templ & \frac{1}{2} \left(y^{1} + y^{2} \right) - 1 & \tempr & 0 & 0 \\ \cline{2-3}
0 & \templ  & - 1 & \tempr & 0 \\ \cline{3-5}
0 & 0  & \templ & - y^{3} + 2 y^{4} & 1 \\
0 & 0 & \templ & - \left( y^{3} - 2 y^{4} \right)^{2} & y^{3} - 2 y^{4} \\
\end{array}\right],
\end{equation}
\begin{equation}
\bs{L}_{2}^{(\bs{y})} =
\left[\begin{array}{ccccc}
y^{1} + ( y^{3} - y^{4} ) y^{4} - y^{5} + 1 & \tempr & 0 & 0 & 0 \\ \cline{1-2}
\templ & y^{1} + ( y^{3} - y^{4} ) y^{4} - y^{5} - 1 & \tempr & 0 & 0 \\ \cline{2-3}
0 & \templ  & y^{3} - y^{4} - 1 & \tempr & 0 \\ \cline{3-5}
0 & 0  & \templ & y^{4} & 1 \\
0 & 0 & \templ & - \left( y^{3} - 2 y^{4} \right)^{2} & 2 y^{3} - 3 y^{4} \\
\end{array}\right] \ ,%
\end{equation}
\begin{equation}
\bs{L}_{3}^{(\bs{y})} =
\left[\begin{array}{ccccc}
1 & \tempr & 0 & 0 & 0 \\ \cline{1-2}
\templ & - 1 & \tempr & 0 & 0 \\ \cline{2-3}
0 & \templ  & - \frac{1}{2} (y^{1} + y^{2}) + y^{5} - 1 & \tempr & 0 \\ \cline{3-5}
0 & 0  & \templ & - \frac{1}{2} (y^{1} + y^{2}) - y^{3} + 2 y^{4} + y^{5} & 1 \\
0 & 0 & \templ & - \left( y^{3} - 2 y^{4} \right)^{2} & - \frac{1}{2} (y^{1} + y^{2}) + y^{3} - 2 y^{4} + y^{5} \\
\end{array}\right] \ .
\end{equation}
Observe that more generally, in these coordinates, the full class \eqref{eq:LTA} takes the block-diagonal form
\begin{equation}
\bs{L}^{(\bs{y})} =
\left[\begin{array}{ccccc}
f_{1} + 1 & \tempr & 0 & 0 & 0 \\ \cline{1-2}
\templ & f_{1} - 1 & \tempr & 0 & 0 \\ \cline{2-3}
0 & \templ  & f_{2} - 1 & \tempr & 0 \\ \cline{3-5}
0 & 0  & \templ & f_{2} + f_{3} - (y^{3} - y^{4}) f'_{3} & 1 \\
0 & 0 & \templ & - \left( f_{3} - (y^{3} - y^{4}) f'_{3} \right)^{2} & f_{2} - f_{3} + (y^{3} - y^{4}) f'_{3} \\
\end{array}\right],
\end{equation}
where $f_{1} = f_{1} \left(\frac{1}{2} (y^{1} + y^{2}), y^{1} + (y^{3} - y^{4}) f_{3} - y^{5} \right)$, $f_{2} = f_{2} \left(y^{4}, - \frac{1}{2} (y^{1} + y^{2}) - (y^{3} - y^{4}) f_{3} + y^{5}, y^{3} - y^{4} \right)$ and $f_{3} = f_{3} (y^{4})$ are arbitrary functions (with $f_{3}' \neq 0$).

\section{A level-four generalized Haantjes algebra}

Let us consider the following fourth-level generalized Nijenhuis family of operator fields defined over a seven-dimensional differentiable manifold $M$, which is obtained with the same procedure described in the previous section:

\begin{equation} \label{eq:LFA1}
\boldsymbol{K}^{(\bs{x})} =
\left[\begin{array}{ccccccc}
g_{1} & x^{1} & x^{1} + g_{2} & - x^{1} & x^{1} & -x^{1} & x^{1} \\
1 & g_{5} & - g_{1} + g_{5} + x^{1} & g_{4} - g_{5} & 0 & g_{1} - g_{5} - x^{1} & - 1 - g_{4} + g_{5} \\
0 & 0 & g_{1} & 0 & 0 & 0 & 0 \\
1 + g_{1} - g_{4} & x^{1} & x^{1} + g_{2} & - 1 - x^{1} + g_{4} & x^{1} & -x^{1} & x^{1} \\
0 & g_{1} - g_{5} & g_{1} + g_{3} - g_{5} + \frac{1}{x^{1}} & - g_{1} + g_{5} & g_{1} & - g_{1} + g_{5} - \frac{1}{x^{1}} & g_{1} - g_{5} \\
0 & 0 & x^{1} & 0 & 0 & g_{1} - x^{1} & 0 \\
g_{1} - g_{4} & x^{1} & x^{1} + g_{2} & - 1 - x^{1} & x^{1} & -x^{1} & 1 + x^{1} + g_{4} \\
\end{array}\right].
\end{equation}

Here $g_{1} = g_{1} ( x^{1}, x^{2}, x^{3}, x^{4} )$,  $g_{2} = g_{2} ( x^{1}, x^{2}, x^{3}, x^{4} )$, $g_{3} = g_{3} ( x^{1}, x^{2}, x^{3}, x^{4} )$, $g_{4} = g_{4} ( x^{5}, x^{6} )$, $g_{5} = g_{5} ( x^{7} )$. 

As a result of a direct computation, one can prove the following 
\begin{proposition}
The operator family \eqref{eq:LFA1} satisfies the  property:
\[
\tau^{(4)}(\bs{K})=0 \ .
\]
In addition, its third level torsion does not vanish whenever  $g_{3} \neq 0$.
\end{proposition}

\subsection{Spectral analysis}
Its minimal polynomial is of seventh degree,
\begin{equation}
m_{\boldsymbol{K}} (\boldsymbol{x}, \lambda) =  ( \lambda - g_{5} )(\lambda-g_4-1)(\lambda-g_4+1) (\lambda - g_{1} + x^{1})(\lambda - g_{1})^{3},
\end{equation}
and its eigenvalues and generalized eigenvectors are
\begin{equation}
\begin{aligned}
\lambda_{1} = g_{5}, \quad \rho_{1} = 1 & \Longrightarrow \cD_{1} = \bigg\langle \dfrac{\partial}{\partial x^{2}} - \dfrac{\partial}{\partial x^{5}} \bigg\rangle, \\
\lambda_{2} = g_{4} + 1, \quad \rho_{2} = 1 & \Longrightarrow \cD_{2} = \bigg\langle \dfrac{\partial}{\partial x^{2}} - \dfrac{\partial}{\partial x^{7}} \bigg\rangle, \\
\lambda_{3} = g_{4} - 1, \quad \rho_{3} = 1 & \Longrightarrow \cD_{3} = \bigg\langle \dfrac{\partial}{\partial x^{2}} + 2 \dfrac{\partial}{\partial x^{4}} + \dfrac{\partial}{\partial x^{7}} \bigg\rangle, \\
\lambda_{4} = g_{1} - x^{1}, \quad \rho_{4} = 1 & \Longrightarrow \cD_{4} = \bigg\langle \dfrac{\partial}{\partial x^{1}} - (x^{1})^{2} \dfrac{\partial}{\partial x^{2}} + \dfrac{\partial}{\partial x^{4}} - \dfrac{\partial}{\partial x^{5}} - (x^{1})^2 \dfrac{\partial}{\partial x^{6}} + \dfrac{\partial}{\partial x^{7}} \bigg\rangle, \\
\lambda_{5} = g_{1}, \quad \rho_{5} = 3 & \Longrightarrow \cD_{5} = \bigg\langle \dfrac{\partial}{\partial x^{1}} + \dfrac{\partial}{\partial x^{4}} + \dfrac{\partial}{\partial x^{7}}, \dfrac{\partial}{\partial x^{5}}, \dfrac{\partial}{\partial x^{3}} + \dfrac{\partial}{\partial x^{6}} \bigg\rangle.
\end{aligned}
\end{equation}
The characteristic distributions and the corresponding annihilators are 
\begin{equation} \label{eq:ann7}
\end{equation}
\begin{align*}
& \mathcal{E}_{1} = \oplus_{i \neq 1} \cD_{i} = \bigg\langle \dfrac{\partial}{\partial x^{2}} - \dfrac{\partial}{\partial x^{7}}, \dfrac{\partial}{\partial x^{2}} + 2 \dfrac{\partial}{\partial x^{4}} + \dfrac{\partial}{\partial x^{7}}, \dfrac{\partial}{\partial x^{1}} - (x^{1})^{2} \dfrac{\partial}{\partial x^{2}} + \dfrac{\partial}{\partial x^{4}} - \dfrac{\partial}{\partial x^{5}} - (x^{1})^2 \dfrac{\partial}{\partial x^{6}} + \dfrac{\partial}{\partial x^{7}}, \\
& \hspace{9cm} \dfrac{\partial}{\partial x^{1}} + \dfrac{\partial}{\partial x^{4}} + \dfrac{\partial}{\partial x^{7}}, \dfrac{\partial}{\partial x^{3}} + \dfrac{\partial}{\partial x^{6}}, \dfrac{\partial}{\partial x^{5}} \bigg\rangle \\
& \hspace{2cm} \Longrightarrow \mathcal{E}_{1}^{\circ} = \langle dx^{2} + dx^{3} - dx^{4} - dx^{6} + dx^{7} \rangle \\
& \mathcal{E}_{2} = \oplus_{i \neq 2} D_{i} = \bigg\langle \dfrac{\partial}{\partial x^{2}} - \dfrac{\partial}{\partial x^{5}}, \dfrac{\partial}{\partial x^{2}} + 2 \dfrac{\partial}{\partial x^{4}} + \dfrac{\partial}{\partial x^{7}}, \dfrac{\partial}{\partial x^{1}} - (x^{1})^{2} \dfrac{\partial}{\partial x^{2}} + \dfrac{\partial}{\partial x^{4}} - \dfrac{\partial}{\partial x^{5}} - (x^{1})^2 \dfrac{\partial}{\partial x^{6}} + \dfrac{\partial}{\partial x^{7}}, \\
& \hspace{9cm} \dfrac{\partial}{\partial x^{1}} + \dfrac{\partial}{\partial x^{4}} + \dfrac{\partial}{\partial x^{7}}, \dfrac{\partial}{\partial x^{3}} + \dfrac{\partial}{\partial x^{6}}, \dfrac{\partial}{\partial x^{5}} \bigg\rangle \\
& \hspace{2cm} \Longrightarrow \mathcal{E}_{2}^{\circ} = \langle dx^{1} + dx^{4} - 2 dx^{7} \rangle \\
& \mathcal{E}_{3} = \oplus_{i \neq 3} \cD_{i} = \bigg\langle \dfrac{\partial}{\partial x^{2}} - \dfrac{\partial}{\partial x^{5}}, \dfrac{\partial}{\partial x^{2}} - \dfrac{\partial}{\partial x^{7}}, \dfrac{\partial}{\partial x^{1}} - (x^{1})^{2} \dfrac{\partial}{\partial x^{2}} + \dfrac{\partial}{\partial x^{4}} - \dfrac{\partial}{\partial x^{5}} - (x^{1})^2 \dfrac{\partial}{\partial x^{6}} + \dfrac{\partial}{\partial x^{7}}, \\
& \hspace{9cm} \dfrac{\partial}{\partial x^{1}} + \dfrac{\partial}{\partial x^{4}} + \dfrac{\partial}{\partial x^{7}}, \dfrac{\partial}{\partial x^{3}} + \dfrac{\partial}{\partial x^{6}}, \dfrac{\partial}{\partial x^{5}} \bigg\rangle \\
& \hspace{2cm} \Longrightarrow \mathcal{E}_{3}^{\circ} = \langle dx^{1} - dx^{4} \rangle, \\
& \mathcal{E}_{4} = \oplus_{i \neq 4} \cD_{i} = \bigg\langle \dfrac{\partial}{\partial x^{2}} - \dfrac{\partial}{\partial x^{5}}, \dfrac{\partial}{\partial x^{2}} - \dfrac{\partial}{\partial x^{7}}, \dfrac{\partial}{\partial x^{2}} + 2 \dfrac{\partial}{\partial x^{4}} + \dfrac{\partial}{\partial x^{7}}, \dfrac{\partial}{\partial x^{1}} + \dfrac{\partial}{\partial x^{4}} + \dfrac{\partial}{\partial x^{7}}, \dfrac{\partial}{\partial x^{3}} + \dfrac{\partial}{\partial x^{6}}, \dfrac{\partial}{\partial x^{5}} \bigg\rangle \\
& \hspace{2cm} \Longrightarrow \mathcal{E}_{4}^{\circ} = \langle dx^{3} - dx^{6} \rangle \\
& \mathcal{E}_{5} = \oplus_{i \neq 5} \cD_{i} = \bigg\langle \dfrac{\partial}{\partial x^{2}} - \dfrac{\partial}{\partial x^{5}}, \dfrac{\partial}{\partial x^{2}} - \dfrac{\partial}{\partial x^{7}}, \dfrac{\partial}{\partial x^{2}} + 2 \dfrac{\partial}{\partial x^{4}} + \dfrac{\partial}{\partial x^{7}}, \\
& \hspace{5cm} \dfrac{\partial}{\partial x^{1}} - (x^{1})^{2} \dfrac{\partial}{\partial x^{2}} + \dfrac{\partial}{\partial x^{4}} - \dfrac{\partial}{\partial x^{5}} - (x^{1})^2 \dfrac{\partial}{\partial x^{6}} + \dfrac{\partial}{\partial x^{7}} \bigg\rangle \\
& \hspace{2cm} \Longrightarrow \mathcal{E}_{5}^{\circ} = \langle dx^{3}, (x^{1})^{2} dx^{1} + dx^{6},  dx^{1} + dx^{2} - dx^{4} + dx^{5} - dx^{6} + dx^{7} \rangle;
\end{align*}
By specializing adequately the arbitrary functions, and maintaining the same choice for the function $g_{3}$, we easily obtain an Abelian,  cyclic and four-level Haantjes algebra. As in the previous example, all the operators of the family obtained in this way will share both the same eigen-distributions and the degree of the minimal polynomial. Thus, any of them generates the full algebra.

\vspace{2mm}

To illustrate the simultaneous diagonalization procedure, by way of an example we can consider the three commuting operators
\begin{equation}
\bs{K}_{1}^{(\bs{x})} =
\left[\begin{array}{ccccccc}
x^{2} & x^{1} & x^{1} & - x^{1} & x^{1} & -x^{1} & x^{1} \\
1 & 0 & x^{1} - x^{2} & 0 & 0 & - x^{1} + x^{2} & - 1 \\
0 & 0 & x^{2} & 0 & 0 & 0 & 0 \\
1 + x^{2} & x^{1} & x^{1} & - 1 - x^{1} & x^{1} & -x^{1} & x^{1} \\
0 & x^{2} & x^{2} & - x^{2} & x^{2} & - \frac{1}{x^{1}} - x^{2} & x^{2} \\
0 & 0 & x^{1} & 0 & 0 & - x^{1} + x^{2} & 0 \\
x^{2} & x^{1} & x^{1} & - 1 - x^{1} & x^{1} & -x^{1} & 1 + x^{1} \\
\end{array}\right] \ ,
\end{equation}
\begin{equation}
\bs{K}_{2}^{(\bs{x})} =
\left[\begin{array}{ccccccc}
x^{3} & x^{1} & - \frac{1}{x^{1}} & - x^{1} & x^{1} & -x^{1} & x^{1} \\
1 & x^{7} & x^{1} - x^{3} + x^{7} & x^{6} - x^{7} & 0 & - x^{1} + x^{3} - x^{7} & - 1 - x^{6} + x^{7} \\
0 & 0 & x^{3} & 0 & 0 & 0 & 0 \\
1 + x^{3} - x^{6} & x^{1} & - \frac{1}{x^{1}} & - 1 - x^{1} + x^{6} & x^{1} & -x^{1} & x^{1} \\
0 & x^{3} - x^{7} & x^{3} - x^{7} & - x^{3} + x^{7} & x^{3} & - \frac{1}{x^{1}} - x^{3} + x^{7} & x^{3} - x^{7} \\
0 & 0 & x^{1} & 0 & 0 & - x^{1} + x^{3} & 0 \\
x^{3} - x^{6} & x^{1} & - \frac{1}{x^{1}} & - 1 - x^{1} & x^{1} & -x^{1} & 1 + x^{1} + x^{6} \\
\end{array}\right] \ ,
\end{equation}
\begin{equation}
 \bs{K}_{3}^{(\bs{x})} =
\left[\begin{array}{ccccccc}
x^{4} & x^{1} & 0 & - x^{1} & x^{1} & -x^{1} & x^{1} \\
1 & 1 & 1 + x^{1} - x^{4} & -1 + x^{5} & 0 & - 1 - x^{1} + x^{4} & - x^{5} \\
0 & 0 & x^{4} & 0 & 0 & 0 & 0 \\
1 + x^{4} - x^{5} & x^{1} & 0 & - 1 - x^{1} + x^{5} & x^{1} & -x^{1} & x^{1} \\
0 & - 1 + x^{4} & - 1 + x^{4} & 1 - x^{4} & x^{4} & 1 - \frac{1}{x^{1}} - x^{4} & - 1 + x^{4} \\
0 & 0 & x^{1} & 0 & 0 & - x^{1} + x^{4} & 0 \\
x^{4} - x^{5} & x^{1} & 0 & - 1 - x^{1} & x^{1} & -x^{1} & 1 + x^{1} + x^{5} \\
\end{array}\right].
\end{equation}
These operators are $C^{\infty}(\mathbb{R})$-linearly independent. Each of them generates cyclically a level-four Haantjes  algebra $\mathscr{H}^{(4)}$ of rank 7. A basis of this algebra is given, for instance, by $\mathscr{B}=\{\bs{I}, \bs{K}_{1}, \bs{K}_{1}^2,\bs{K}_{1}^3 ,\bs{K}_{1}^4,\bs{K}_{1}^5, \bs{K}_{1}^{6} \}$. In the following subsection, we shall determine a set of coordinates which block-diagonalize the full algebra $\mathscr{H}^{(4)}$.

\subsection{Block-diagonalization} By integrating the annihilators \eqref{eq:ann7} of the eigendistributions of these operators, we get the block-separating coordinates:
\begin{equation}
\begin{aligned}
& y^{1} = x^{2} + x^{3} - x^{4} - x^{6} + x^{7}, \\
& y^{2} = x^{1} + x^{4} - 2 x^{7}, \\
& y^{3} = x^{1} - x^{4}, \\
& y^{4} = x^{3} - x^{6}, \\
& y^{5} = x^{3}, \\
& y^{6} = \frac{(x^{1})^{3}}{3} - x^{6}, \\
& y^{7} = x^{1} + x^{2} - x^{4} + x^{5} - x^{6} + x^{7}.
\end{aligned}
\end{equation}


In these coordinates, all the operators of the algebra $\mathscr{H}^{(4)}$ block-diagonalize simultaneously. In particular, by defining $\chi = \sqrt[3]{3 (y^{4} - y^{5} + y^{6})}$, the operators $\bs{K}_1$, $\bs{K}_2$, $\bs{K}_3$ take the form:
\begin{scriptsize}
\begin{equation}
\bs{K}_{1}^{(\bs{y})} =
\left[\begin{array}{ccccccc}
0 & \tempr & 0 & 0 & 0 & 0 & 0 \\ \cline{1-2}
\templ & 1 & \tempr & 0 & 0 & 0 & 0 \\ \cline{2-3}
0 & \templ  & - 1 & \tempr & 0 & 0 & 0 \\ \cline{3-4}
0 & 0  & \templ & y^{1} + \frac{1}{2} (y^{2} - y^{3}) - y^{4} - \chi & \tempr & 0 & 0  \\ \cline{4-7}
0 & 0 & 0 & \templ & y^{1} + \frac{1}{2} (y^{2} - y^{3}) - y^{4} & 0 & 0 \\
0 & 0 & 0 & \templ & \chi + \chi^{3} & y^{1} + \frac{1}{2} (y^{2} - y^{3}) - y^{4} - \chi & \chi^{3} \\
0 & 0 & 0 & \templ & \chi & - \chi^{-1}  & y^{1} + \frac{1}{2} (y^{2} - y^{3}) - y^{4} + \chi \\
\end{array}\right] \ ,
\end{equation}
\end{scriptsize}
\begin{equation}
\bs{K}_{2}^{(\bs{y})} =
\left[\begin{array}{ccccccc}
- \frac{1}{2} (y^{2} + y^{3}) + \chi & \tempr & 0 & 0 & 0 & 0 & 0 \\ \cline{1-2}
\templ & 1 - y^{4} + y^{5} & \tempr & 0 & 0 & 0 & 0 \\ \cline{2-3}
0 & \templ  & - 1 - y^{4} + y^{5} & \tempr & 0 & 0 & 0 \\ \cline{3-4}
0 & 0  & \templ & y^{5} - \chi & \tempr & 0 & 0  \\ \cline{4-7}
0 & 0 & 0 & \templ & y^{5} & 0 & 0 \\
0 & 0 & 0 & \templ & 0 & y^{5} - \chi & \chi^{3} \\
0 & 0 & 0 & \templ & - \chi^{-1} & - \chi^{-1}  & y^{5} + \chi \\
\end{array}\right] \ ,
\end{equation}
\begin{equation}
\bs{K}_{3}^{(\bs{y})} =
\left[\begin{array}{ccccccc}
1 & \tempr & 0 & 0 & 0 & 0 & 0 \\ \cline{1-2}
\templ & 1 - y^{1} + y^{5} - \chi + y^{7} & \tempr & 0 & 0 & 0 & 0 \\ \cline{2-3}
0 & \templ  & - 1 - y^{1} + y^{5} - \chi + y^{7} & \tempr & 0 & 0 & 0 \\ \cline{3-4}
0 & 0  & \templ & - y^{3} & \tempr & 0 & 0  \\ \cline{4-7}
0 & 0 & 0 & \templ & - y^{3} + \chi & 0 & 0 \\
0 & 0 & 0 & \templ & \chi & -y^{3} & \chi^{3} \\
0 & 0 & 0 & \templ & 0 & - \chi^{-1}  & - y^{3} + 2 \chi \\
\end{array}\right].
\end{equation}

Also, in the block-separation coordinates, 
the full family \eqref{eq:LFA1} takes a block diagonal form:
\begin{equation}
\bs{K}^{(\bs{y})}=
\left[\begin{array}{ccccccc}
g_{5} & \tempr & 0 & 0 & 0 & 0 & 0 \\ \cline{1-2}
\templ & g_{4} + 1 & \tempr & 0 & 0 & 0 & 0 \\ \cline{2-3}
0 & \templ  & g_{4} - 1 & \tempr & 0 & 0 & 0 \\ \cline{3-4}
0 & 0  & \templ & g_{1} - \chi & \tempr & 0 & 0  \\ \cline{4-7}
0 & 0 & 0 & \templ & g_{1} & 0 & 0 \\
0 & 0 & 0 & \templ & g_{2} \chi^{2} + \chi \left( 1 + \chi^{2} \right) & g_{1} - \chi & \chi^{3} \\
0 & 0 & 0 & \templ & g_{3} + g_{2} + \chi^{-1} \left( 1 + \chi^{2} \right) & - \chi^{-1}  & g_{1} + \chi \\
\end{array}\right].
\end{equation}

Here $g_{1} = g_{1} ( \chi, y^{1} + \frac{1}{2} (y^{2} - y^{3}) - y^{4}, y^{5} )$,  $g_{2} = g_{2} ( \chi, y^{1} + \frac{1}{2} (y^{2} - y^{3}) - y^{4}, y^{5} )$, $g_{3} = g_{3} ( \chi, y^{1} + \frac{1}{2} (y^{2} - y^{3}) - y^{4}, y^{5} )$, $g_{4} = g_{4} ( - y^{1} -  \chi + y^{5} + y^{7}, - y^{4} + y^{5} )$ and $g_{5} = g_{5} ( - \frac{1}{2} (y^{2} + y^{3})+ \chi )$ are arbitrary functions.


\section*{Acknowledgement}


This work has been partly supported by the research project PGC2018-094898-B-I00, MICINN, Spain, and by the ICMAT Severo Ochoa project SEV-2015-0554 (MICINN). D. R. N. acknowledges the financial support of EXINA S.L.

P. T. and G. T. are members of the Gruppo Nazionale di Fisica Matematica (GNFM).


\begin{thebibliography}{00} %


\bibitem{BogCMP} O. I. Bogoyavlenskij.  Necessary conditions for existence of non-degenerate Hamiltonian structures. \textit{Commun.  Math. Phys.} \textbf{182}, 253--290 (1996).

\bibitem{BogIzv2004} O.I. Bogoyavlenskij.  General algebraic identities for the Nijenhuis and Haantjes torsions. \textit{Izvestya Mathematics}  \textbf{68}, 1129-1141 (2004).

\bibitem{BogJMP} O. I. Bogoyavlenskij. Block-diagonalizability problem for hydrodynamic type systems. \textit{J. Math. Phys.} \tbf{47}, paper 063502, 9 pages (2006).

\bibitem{BogCMP2007} O. I. Bogoyavlenskij. Decoupling problem for systems of quasi-linear PDE's. \textit{Commun. Math. Phys.} \tbf{269}, 545-556 (2007). 


\bibitem{BKM2022AIM} A. V. Bolsinov, A. Yu. Konyaev, V. S. Matveev. Nijenhuis geometry. \textit{Adv. Math.} \tbf{394}, 108001 (2022).


\bibitem{CR2019} C. M. Chanu, G. Rastelli. Block-separation of variables: a form of partial separation for natural Hamiltonians.
\textit{SIGMA}, \textbf{15}, paper 013, 22 pages (2019)

\bibitem{CH1962} R. Courant and D. Hilbert. \textit{Methods of Mathematical Physics}, II, Interscience Publishers, New York, 1962.

\bibitem{FeMa} E. V. Ferapontov and D. G. Marshall. Differential-geometric approach to the integrability of hydrodynamic chains: the Haantjes tensor. \textit{Math. Annalen} \textbf{339}, 61--99 (2007).

\bibitem{FeKhu} E. Ferapontov, K. Khusnutdinova. The Haantjes tensor and double waves for multi-dimensional systems of hydrodynamic type: a necessary condition for integrability. \textit{Proc. Royal Soc. A}, \tbf{462}, 1197-1219 (2006).

\bibitem{FP2021} E. V. Ferapontov and M. V. Pavlov. Kinetic equation for soliton gas: integrable reductions, arxiv:2109.11962v1 (2021).

\bibitem{FN1956} A. Fr\"olicher, A. Nijenhuis. Theory of vector valued differential forms. Part I. \textit{Indag. Math.} \textbf{18},  338--359 (1956).


\bibitem{Haa1955} J. Haantjes. On $X_{n-1}$-forming sets of eigenvectors. \textit{Indag. Math.} \textbf{ 17}, 158--162 (1955).



\bibitem{KS2019} Y. Kosmann-Schwarzbach. Beyond recursion operators, Preprint arXiv:1712.08908, 2017,  in \textit{Proceedings of the  XXXVI Workshop on Geometric Methods in Physics}, Bia{\l}owie\.{z}a, Poland, July 2017, Birkhauser, (2019).



%
%
\bibitem{MGall17} F. Magri. Haantjes manifolds with symmetries.  \textit{Theor. Math. Phys.}, \textbf{196}, 1217--1229 (2018).


\bibitem{NNAM1957} A. Newlander, L. Nirenberg. Complex analytic coordinates in almost complex manifolds. \textit{Ann. of Math.}  \textbf{65}, 391--404 (1957).

\bibitem{Nij1951} A. Nijenhuis. $X_{n-1}$-forming sets of eigenvectors. \textit{Indag. Math.}  \textbf{54}, 200--212 (1951).

\bibitem{Nij1955} A. Nijenhuis. Jacobi-type identities for bilinear differential concomitants of certain tensor fields. I, \textit{Indag. Math.} \textbf{17}, 390--397 (1955).

\bibitem{Nij19552} A. Nijenhuis. Jacobi-type identities for bilinear differential concomitants of certain tensor fields. II. \textit{Indag. Math.} \textbf{17}, 398--403 (1955).

\bibitem{RTT2022CNS} D. Reyes Nozaleda, P. Tempesta and G. Tondo. Classical multiseparable Hamiltonian systems, superintegrability and Haantjes geometry. \textit{Commun. Nonl. Sci. Num. Sim.} \tbf{104}, 106021  (2022).
%
 
\bibitem{TT2021JGP} P. Tempesta, G. Tondo. Haantjes Algebras and diagonalization (preprint, arXiv:1710.04522) \textit{J. Geom. Phys.} \tbf{160}, paper 103968, 21 pages (2021). 

\bibitem{TT2022AMPA} P. Tempesta, G. Tondo. Haantjes manifolds and classical integrable systems. \textit{Ann. Mat. Pura Appl.} \tbf{201}, 57–90 (2022).



\bibitem{TT2022CMP} P. Tempesta, G. Tondo. Higher Haantjes brackets and integrability.  \textit{Commun. Math. Phys. \tbf{389}}, 1647–1671 (2022).

\bibitem{TT2016SIGMA}  G. Tondo, P. Tempesta. Haantjes structures for the Jacobi-Calogero model and the Benenti Systems. \textit{SIGMA} \textbf{12}, paper 023, 18 pages (2016).

\bibitem{T2018TMP}  G. Tondo. Haantjes algebras of the Lagrange top. \textit{Theor. Math. Phys.} \textbf{196}, 1366--1379 (2018).

\end{thebibliography}
\end{document}